\documentclass [preprint,12pt,english] {amsart}
\usepackage[usenames]{color}
\usepackage{graphicx}
\usepackage{natbib}
\usepackage{lineno}
\usepackage{setspace}
\usepackage{enumitem}
\newtheorem{proposition}{Proposition}[section]

\newtheorem{theorem}{Theorem}[section]
\def\a{\alpha}

\newcommand{\si}{\sigma}

\newcommand{\ta}{\tau}
\def\th{\theta}

\def\R{\mathbb R}

\newcommand{\ab}{{\bar a}}
\newcommand{\xb}{{\bar x}}
\newcommand\Wb{\,\overline{\!W\!}\,}

\begin{document}
\title[Evolutionary altered outcomes]{Evolutionarily induced alternative states and coexistence in systems with apparent competition}
\author[S.J. Schreiber]{Sebastian J. Schreiber}
\author[S. Patel]{Swati Patel}
\address{Department of Evolution and Ecology, One Shields Avenue, University of California, Davis, California 95616}
\email{sschreiber@ucdavis.edu}
\email{swpatel@ucdavis.edu}
\maketitle
\bibliographystyle{plainnat}

\begin{abstract}
Predators often consume multiple prey and by mutually subsidizing a shared predator, the prey may reciprocally harm each other. When predation levels are high, this apparent competition can culminate in a prey species being displaced. Coupling quantitative genetics and Lotka-Volterra models, we study how predator evolution alters this and other ecological outcomes. These models account for a trade-off between the predator's attack rates on two prey species. We provide a mathematical characterization of a strong form of persistence--permanence--for which there is a global attractor bounded away from extinction. When the evolutionary dynamics occur at a sufficiently slower time scale than the ecological dynamics, we also characterize attractors and their basins' of attraction using singular perturbation theory and a graphical approach to the eco-evolutionary dynamics. Our results show that eco-evolutionary feedbacks can mediate permanence at intermediate trade-offs in the attack rates. However, at strong trade-offs, permanence is lost. Despite this loss of permanence, there can be attractors supporting coexistence. These attractors, however, may coincide with attractors at which the predator is excluded. Our results highlight that evo-evolutionary feedbacks can alter community structure by mediating coexistence or leading to trait-dependent alternative stable states.  
\end{abstract}

\bibliographystyle{plainnat}

\section{Introduction}

Ecological communities consist of complex webs of interacting species, each of which contain phenotypically diverse individuals. Species interactions including competition, predation, and mutualism, generate nonlinear feedbacks that determine community composition and stability. These feedbacks are the key focus of community ecology theory. Within each species, individuals often differ in many traits including gender, size, behavior, or physiology. This variation provides the raw material for natural selection and thus is a key focus of evolutionary theory. Traditionally, evolutionary biologists and community ecologists developed theory and ran experiments without considering the other discipline~\citep{fussmann-etal-07}. This separation  stemmed from the traditional belief that ecological and evolutionary processes occur on vastly different time scales. However, in \citet{schoener-11}'s review of ``the newest synthesis'', there is growing empirical evidence that feedbacks between ecological and evolutionary processes occur on more commensurate time scales (e.g. tens to thousands of generations instead of hundreds of thousands of generations) and the effects of these feedbacks can be substantial. Hence, ``[n]othing in evolution or ecology [may] make sense except in the light of the other''~\citep{pelletier-etal-09}.  A major challenge facing this synthesis, ``is  whether the persistence of interactions and the stability of communities truly rely upon ongoing rapid evolution or whether such rapid evolution is ecologically trivial''~ \citet{thompson-99}. Here we confront this challenge for the ``apparent competition'' community module in which two prey species share a common predator species. 

Predators often consume multiple prey. By mutually subsidizing a shared predator, the prey species may reciprocally harm each other and, thereby appear to be competing~\citep{holt-77,holt-lawton-93,holt-lawton-94}. These negatively reciprocal responses have been demonstrated empirically in plant-herbivore systems~\citep{rand-03,rand-etal-04}, insect communities~\citep{mueller-godfray-97,rott-godfray-98,morris-etal-01}, and hosts sharing common pathogens~\citep{tompkins-etal-00,cobb-etal-10}. When prey are not resource limited or predation pressure is strong, models predict that  apparent competition can lead to the exclusion of one of the prey species -- dynamic monophagy~\citep{holt-lawton-94}. 

From the perspective of the predator, attacking multiple prey may ensure greater energetic gains or provide insurance against the loss of a focal prey species~\citep{macarthur-55}. To what extent these benefits exist depends largely on the trade-offs between the abilities of attacking the different prey species. For example, many predator species optimal attack rate occur at an intermediate ratio of predator to prey body sizes~\citep{brose-etal-06,brose-10}. Hence, predators experience a size trade-off between being larger to optimize attack rates on larger prey and being smaller to optimize attack rates on smaller prey.      

Early work on the evolution of a predator attacking two prey species has focused on the consequences of the ecological and evolutionary feedbacks on the phenotypic distribution of the predator~\citep{wilson-turelli-86,jmb-03,rueffler-etal-06,abrams-06a,abrams-06b,nurmi-parvinen-13}. For example, using a single-locus selection model based on differential utilization of two prey species, \citet{wilson-turelli-86} illustrated that there can be selection for polymorphic predators in which, surprisingly, the heterozygous individuals are the least fit. More recently, \citet{ecology-11b} used a quantitative genetics framework to examine how phenotypic variation in a predator affects ecological outcomes. They found that eco-evo feedbacks can marginalize or even reverse the negative effects of apparent competition and mediate coexistence of the prey species. Furthermore, at sufficiently strong trade-offs, the eco-evolutionary dynamics exhibit alternative stable states. This study, however, did not examine how trade-offs influence predator persistence or provide mathematically rigorous proofs of their results. 

Here, we provide a mathematically rigorous verification of the results of \citet{ecology-11b} and explore how trade-offs influence predator as well as prey persistence.  In section 2, we introduce the model which couples Lotka-Volterra dynamics for apparent competition with a quantitative genetics model for the predator trait. In section 3, we study a strong form of persistence, namely permanence at which there is a global attractor bounded away from extinction~\citep{hutson-schmitt-92}. We characterize permanence and examine how the strength of trade-offs influences permanence. In section 4, we refine our analysis of the eco-evolutionary dynamics using singular perturbation techniques and provide estimates for the size of the basins of attraction for the stable equilibria of eco-evolutionary dynamics. Coupling these results with a graphical approach to the eco-evolutionary dynamics, we explore the attractor structure of the eco-evolutionary dynamics and identify under what conditions there is conditional coexistence of all three species. In section 5, we conclude with a discussion of the main implications of our work and highlight future research directions.  The proofs of the main results are presented in sections 6 and 7.

\section{Coupling the Ecological and Evolutionary Dynamics}

There exist a variety of ways that modelers have coupled ecological and evolutionary processes. These approaches differ in whether the traits under selection are discrete or continuous, the underlying genetic architecture and processes (e.g. clonal evolution of haploid individuals versus diploid individuals with recombination), and whether the entire trait distribution is modelled. Here, we use a quantitative genetics framework which assumes a continuous trait and bears some similarities to the adaptive dynamics approach. We develop the model from~\citet{ecology-11b} step by step to illustrate how the quantitative genetics framework can be used to develop models with eco-evo feedbacks, and to highlight the assumptions underlying these models. 

We begin by describing the ecological dynamics in the absence of evolution. These dynamics involve a predator population with density $P$ consuming two prey species with densities $N_1$ and $N_2$. Each prey species $i$ exhibits logistic dynamics with intrinsic rate of growth $r_i$ and carrying capacity $K_i$. The attack rate of the predator on prey species $i$ is $a_i$. The conversion efficiency $e_i$ determines how the number of prey $i$ eaten by the predator convert to new predator numbers, and $d$ is the predator's per-capita death rate. Hence, the ecological dynamics are given by
\begin{equation}\label{eco}
\begin{aligned}
\frac{d N_i}{dt} &= r_iN_i(1-N_i/K_i) - P a_i N_i \\
\frac{dP}{dt} &= P (e_1a_1N_1+e_1a_1N_1-d)
\end{aligned}
\end{equation}
\citet{takeuchi-adachi-83} proved that this Lotka-Volterra system always has a globally stable equilibrium that either supports all three species (coexistence), both prey species (predator exclusion), or the predator and one prey species (prey exclusion). These outcomes can be determined by examining the per-capita growth rates of missing species from the boundary equilibria. 

To overlay the evolutionary dynamics on top of the ecological dynamics, we take a quantitative genetics approach, sometimes playfully called ``Lande Land''~\citep{lande-76}. In Lande Land, each individual predator's phenotype is determined by a continuous trait $x$. This trait determines the predator's attack rate $a_i(x)$ on each of the prey species $i=1,2$. The attack rates $a_i(x)$  are maximal at an optimal trait value $x=\theta_i$ and decrease away from this optimal trait value in a Gaussian manner, i.e., $a_i(x) = \a_i \exp\Bigl[-\frac{(x-\th_i)^2}{2\ta_i^2}\Bigr]$, where $\alpha_i$ is the maximal attack rate and $\tau_i$ determines how steeply attack rate declines with distance from the optimal trait value. In effect, $\tau_i$ determines how phenotypically specialized a predator must be to use prey $i$. 

This model of the attack rates mimics the common empirical situation in which quantitative trait variation in a predator influences individuals relative use of alternative resources. For predator-prey interactions ranging from terrestrial predators of arthropods to aquatic predators of zooplankton~\citep{brose-10}, the attack rate of a predator is optimal at intermediate predator-prey body size ratios~\citep{brose-etal-06}. When predators attack prey of different sizes, there is an inherent trade-off between being larger to better attack the larger prey species and being smaller to better attack the smaller prey species. Alternatively, individuals of threespine stickleback (\emph{Gasterosteus aculeatus}) preferentially consume either benthic insect larvae or limnetic zooplankton~\citep{araujo-etal-08, matthews-etal-10}. For these populations, gill raker length or number are approximately normally distributed, and individuals at different ends of this distribution tend to consume different prey types~\citep{robinson-00}.

In Lande Land, there are two main assumptions. First, the trait remains normally distributed in the population~\citep{lande-76,turelli-barton-94}. This assumption corresponds to the trait being determined by additive contributions of many independent loci. Second, the variance $\si^2$ in this trait remains constant over time.  This assumption allows for moment closure: when it holds, the mean trait value determines the entire distribution. While this assumption is likely to be violated over longer-time scales, it provides an analytically tractable, first approximation to the full distributional dynamics. While epistasis, linkage disequilibrium or
genotype-by-environment interaction can generate substantial deviations
from a normal trait distribution, \citet{turelli-barton-94} showed
numerically that the normal approximation still gives remarkably accurate
predictions for dynamics of the mean and variance of the trait value under
a wide variety of assumptions. Under weak assumptions, even
frequency-dependent disruptive selection maintains a nearly Gaussian trait
distribution~\citep{burger-gimelfarb-04}. 

Let $\xb$ and $\si^2$ be the mean and variance of this normally distributed trait. The phenotypic variance $\si^2=\si_G^2+\si_E^2$ has a genetic and an environmental component. The environmental component $\si_E^2$ corresponds to non-heritable variation in the trait e.g. individuals that developed in different environmental backgrounds. The phenotypic variation in the trait, whether heritable or not, influences the ecological dynamics due to nonlinear averaging of the attack rate at the scale of the predator population. Specifically, the average attack rate on prey $i$ is
\[
\ab_i(\xb) = \int_{-\infty}^\infty a_i(x) p(x,\xb)\, dx
= \frac{\a_i \ta_i}{\sqrt{\si^2+\ta_i^2}} \exp\Bigl[-\frac{(\xb-\th_i)^2}{2(\si^2+\ta_i^2)}\Bigr]\,.
\]
where $p(x,\xb) = \frac{1}{\sqrt{2\pi \si^2}} \exp\Bigl[-\frac{(x-\xb)^2}{2\si^2}\Bigr]$,
is the density of the normal distribution with mean $\xb$ and variance $\si^2$. 

Under the assumptions of a normally distributed trait with fixed variance, \citet{lande-76} showed that the rate of change of the mean trait is given by 
\[
\frac{d\xb}{dt}=\sigma_G^2 \frac{\partial \Wb}{\partial \xb} 
\]
where $\Wb$ is the average per-capita growth rate or fitness of the evolving species. Here, 
\[
\Wb=\sum_{i=1}^2 e_i \ab_i(\xb) N_i - d
\]
for the predator. In words, the rate of change of the mean trait is proportional to the genetic variance of the trait and the gradient of the fitness. Intuitively, evolution selects for increasing fitness. However, due to ecological feedbacks, the graph of the fitness function, ``the fitness landscape'', may change as the trait changes, leading to eco-evo feedbacks. 

Putting all the pieces together, we get that the ecological and evolutionary dynamics are given by 
\begin{equation}\label{dyn1}
\begin{aligned}
\frac{d N_i}{dt} &= r_iN_i(1-N_i/K_i) -  \ab_i(\xb) N_i P\\
\frac{dP}{dt} &= P \,\Wb\\
\frac{d\xb}{dt} &= \si_G^2 \frac{\partial \Wb}{\partial \xb} 
\end{aligned}
\end{equation}
where
\begin{equation}\label{dyn2}
\frac{\partial \Wb}{\partial \xb} = \sum_{i=1}^2 \frac{e_i N_i \tau_i \alpha_i (\theta_i -\xb)}{(\tau_i^2+\sigma^2)^{3/2}} \exp \left[-\frac{(\xb-\theta_i)^2}{2(\tau_i^2+\sigma^2)}\right].
\end{equation}
The state space for these dynamics are $\R^3_+\times \R$ where $\R_+=[0,\infty)$. 

\section{Permanence}

We begin by examining the conditions that ensure that all three species coexist in the sense of permanence, and how the strength of the trade-off affects these conditions. We say our system is \emph{permanent} if there exists $\beta>0$ such that \[
1/\beta\ge \limsup_{t \rightarrow \infty} \max\{N_1(t),N_2(t),P(t)\}\ge \liminf_{t \rightarrow \infty} \min\{N_1(t),N_2(t),P(t)\} \geq \beta
\] for all initial positive population densities ($N_1(0)N_2(0)P(0)>0$) and any initial phenotype in $\xb(0) \in \mathbb{R}$.

To evaluate permanence for our eco-evolutionary system, we begin by considering the dynamics in each two-species subsystem. For the predator-prey subsystem $(N_i, P, \xb)$, the following proposition implies that the predator evolves to specialize on the present prey and may or may not persist with that prey at that phenotype. This proposition states the result for $i=1$, but also applies for $i=2$. 

\begin{proposition}\label{prop:pred-prey}
Assume $N_2(0)=0$ and $N_1(0)P(0)>0$. If $e_1\ab_1(\theta_1)>d$, then 
\[
\lim_{t\to\infty} (N_1(t),P(t),\xb(t))= \left(\frac{d}{e_1\ab_1(\theta_1)}, \frac{r_1(1-\frac{\hat{N}_1}{K_1})}{\ab_1(\theta_1)}, \theta_1\right).
\]
Alternatively, if $e_1\ab_1(\theta_1)\leq d$, then 
\[
\lim_{t\to\infty} (N_1(t),P(t),\xb(t))= (K_1,0, \theta_1).
\]
\end{proposition}

For the two prey subsystem, we have the following characterization of the eco-evolutionary dynamics. 

\begin{proposition}\label{prop:prey-prey}
Assume $P(0)=0$ and $N_1(0)N_2(0)>0$.  Then
\[
\lim_{t\to\infty} (N_1(t), N_2(t),\xb(t))=(K_1, K_2, \xb^*)
\]
for some $\xb \in Q=\{\xb\in \R | \frac{\partial{W}}{\partial{\xb}}(K_1, K_2, \xb)=0\}$ 
\end{proposition}

We remark that for the single prey subsystem (i.e. $P(0)=N_j(0)=0, N_i(0)>0; i\neq j$), one can easily verify that 
\[
\lim_{t\to\infty} (N_i(t), \xb(t))=(K_i, \theta_i).
\]

In Proposition~\ref{prop:prey-prey}, $Q$ is the set of fitness extrema when both the prey are at carrying capacity.  We prove both propositions in section 6. By examining the per-capita growth of the absent species in these subsystems, in which either the predator or prey is excluded, we get the following theorem for permanence in our eco-evolutionary system:

\begin{theorem}~\label{thm:permanence} 
Let $Q=\{\xb\in \mathbb{R} | \frac{\partial{W}}{\partial{\xb}}(K_1, K_2, \xb)=0\}$ be the set of fitness extrema when both the prey are at carrying capacity. If 
\begin{enumerate}[label=(\alph*)]
\item $\frac{r_i}{\ab_i(\theta_j)} > \frac{r_j}{\ab_j(\theta_j)} (1-\frac{d}{\ab_j(\theta_j)e_jK_j})$ for $i=1,2; i\neq j$ and 
\item $\bar{W}(K_1, K_2, \xb^*) > 0$ for all $\xb^* \in Q$
\end{enumerate}
 then the system is permanent in $\mathbb{R}_+^3 \times \mathbb{R}$.  

Conversely, if the inequality in condition (a) or (b) is reversed, then the system is not permanent.  In particular, if the inequality in (a) is reversed for $i=1$ ($2$, respectively), then the equilibrium point
$ (\frac{d}{e_1\ab_1(\theta_1)}, 0, \frac{r_1(1-\frac{\hat{N}_1}{K_1})}{\ab_1(\theta_1)}, \theta_1)$
($(0, \frac{d}{e_2\ab_2(\theta_2)}, \frac{r_2(1-\frac{\hat{N}_2}{K_2})}{\ab_2(\theta_2)}, \theta_2)$, respectively)
is stable. If the inequality in condition (b) is reversed for some $\xb^*\in Q$, then there exists initial positive population densities $(N_1(0)N_2(0)P(0)>0)$ and phenotype $\xb(0)$ such that

\[
\lim_{t\rightarrow\infty} (N_1(t),N_2(t),P(t),\xb(t))= (K_1,K_2,0,\xb^*)
\]
\end{theorem}

\paragraph{\bf Note:} We conjecture that in fact the conditions for permanence ensure robust permanence (see, e.g., \citet{jde-00}). 

\begin{figure}
\includegraphics[width=0.6\textwidth]{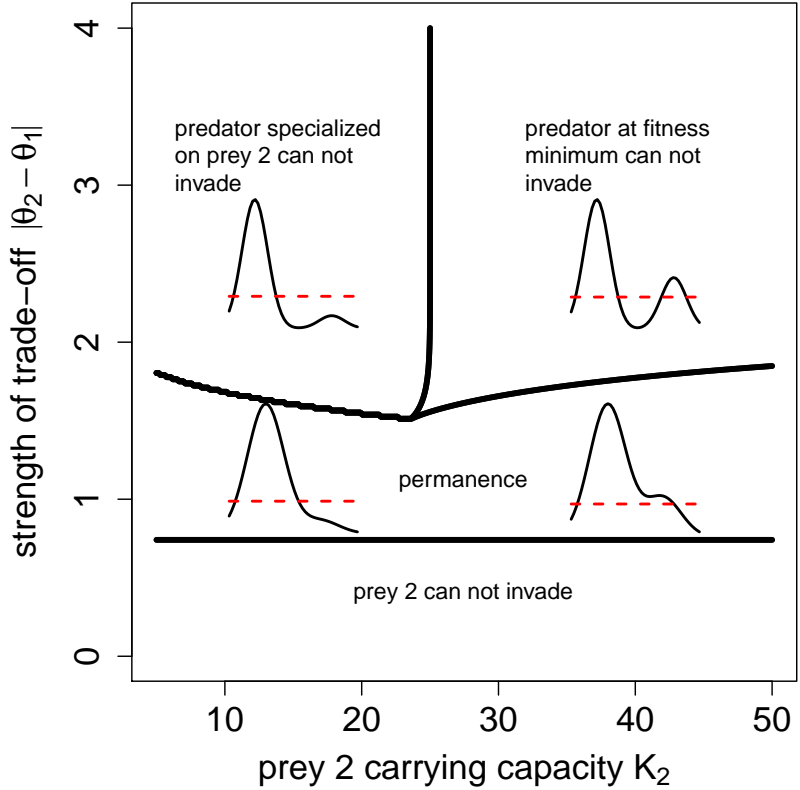}
\caption{The effects of trade-off strengths and carrying capacities on permanence. In this bifurcation diagram, prey $2$ is the inferior prey species (i.e. has the lower value of $r_i/\ab_i(\theta_i)$). At intermediate trade-offs, the system is always permanent. At weak trade-offs, permanence is lost due to the inferior prey species being unable to invade. At strong trade-offs, the predator can only conditionally establish. The graphs in these regions correspond to fitness landscapes $\Wb$ when both prey are at their carrying capacities. Parameters: $r_1=1,r_2=0.25$, $K_1=100$, $\alpha_1=\alpha_2=0.01,\tau_1=\tau_2=0.4,\sigma=0.3, e_1=e_2=0.5$ ,and $d=0.1$.}\label{fig:permanence}
\end{figure}

We can use this theorem to infer how the trade-off and carrying capacities affect permanence (Fig.~\ref{fig:permanence}).  The strength of the trade-off, measured as $|\theta_2-\theta_1|$, has two main effects on permanence, that cumulatively result in permanence only occurring at intermediate trade-offs.  First, strong trade-offs affect prey persistence. 
Let us assume that $\frac{r_1}{\ab_1(\theta_1)}>\frac{r_2}{\ab_2(\theta_2)}$ in which case condition (a) is always met for prey $1$ i.e. prey $1$ always persists, independent of the trade-off.  If the carrying capacity of prey $1$ is sufficiently large and there is no trade-off (i.e. $\theta_1=\theta_2$), then the equilibrium only supporting prey $1$ and the predator is stable and the system is not permanent. However, increasing the trade-off $|\theta_2-\theta_1|$ decreases the attack rate,  $\ab_2(\theta_1)$, on prey $2$ when the predator is specialized on prey $1$. Hence, condition (a) for prey $2$, is more easily satisfied at stronger trade-offs (e.g. prey $2$ persists provided $|\theta_1-\theta_2|>0.75$ in Fig.~\ref{fig:permanence}).  

Second, the strength of the trade-off determines the shape of the fitness landscape and consequently, the set $Q$ of fitness extrema in condition (b).  For  $\tau_1=\tau_2$, we have shown elsewhere~\citep{patel-schreiber-preprint} that $|\theta_2-\theta_1|\leq 2\sqrt{\sigma^2+\tau^2}$ implies that the predator fitness landscape is unimodal with a single fitness maximum (see fitness curves in Fig.~\ref{fig:permanence} for $|\theta_2-\theta_1|\approx 1$). Thus, for weak trade-offs it suffices to check the growth at this optimal phenotype for when the prey are at carrying capacity. Alternatively, for $|\theta_2-\theta_1|> 2\sqrt{\sigma^2+\tau^2}$, the fitness curves can have multiple extrema, typically with two fitness maxima near the specialized phenotypes and one minimum at intermediate phenotypes when the prey are at carrying capacity (see fitness curves in Fig.~\ref{fig:permanence} for $|\theta_2-\theta_1|\approx 2.5$). For sufficiently strong trade-offs, the predator will not have positive growth at intermediate phenotypes, and specifically at the fitness minimum, as it cannot attack either prey efficiently. Hence, the system is not permanent at sufficiently strong trade-offs as there are initial conditions leading to the exclusion of the predator. Furthermore, as shown in Fig.~\ref{fig:permanence}, if one of the prey has a sufficiently low carrying capacity, then there is a stable equilibrium excluding predators experiencing strong trade-offs. However, coexistence is still possible in a weaker sense as we discuss in the next section.

\section{A fast-slow approximation}

To refine our understanding of the eco-evolutionary dynamics, we observe that the phenotypic dynamics occur at a slower time scale than the ecological dynamics. Hence, as a first approximation, we consider the limiting case of when the ecological dynamics are much faster than the evolutionary dynamics. To perform this time scale separation in a quantitative genetics framework requires a bit of care as the genetic variation $\sigma_G^2$ which scales the rate of change $\xb$ in the phenotypic dynamics also influences the ecological dynamics through the phenotypic variation term $\sigma^2=\sigma^2_G+\sigma^2_E$. To separate out these effects, we consider the heritability of the phenotypic variation 
\[
h^2=\frac{\sigma^2_G}{\sigma^2}
\]
which varies between $0$, when none of the phenotypic variation is inherited, and $1$, when all the phenotypic variation is inherited. For a fixed level of phenotypic variation $\sigma^2$, varying $h^2$ only influences the speed of the evolutionary dynamic and has no immediate effect on the ecological dynamics. Hence, for our fast-slow approximations, we assume that $\sigma^2$ is fixed and we vary the speed of evolution by varying $h^2$.

If $h^2=0$, then $\xb$ remains constant and the dynamics of \eqref{dyn1} correspond to the classical Lotka-Volterra dynamics of two prey species with a common predator. \cite{takeuchi-adachi-83} have studied these dynamics in great detail and have shown (cf. Theorem 6 with $\alpha=\beta=0$) that all of the positive solutions of these equations converge to a unique stable equilibrium, call it  $(\hat N_1(\xb), \hat N_2(\xb), \hat P (\xb))$,  whenever all species are initially present. This equilibrium can either support all three species (i.e. all components are positive), only the prey species (i.e. only the first two components are positive), or the predator species and only one of the prey species (i.e. one of the first two components is zero, the other components are positive). The graph of this function, $\mathcal{E}=\{(\hat N_1(\xb), \hat N_2(\xb), \hat P (\xb),\xb)|\xb \in [\theta_1,\theta_2]\}$, defines a piecewise smooth, one dimensional manifold homeomorphic to $[\theta_1,\theta_2]$ that is a global attractor for the dynamics of \eqref{dyn1} when $h=0$  i.e., all solutions with all species initially present converge to $\mathcal{E}$. 

When $h^2$ is positive but small, we can approximate the dynamics of the fully coupled system with the fast-slow approximation:
\begin{equation}\label{fast}
\begin{aligned}
\frac{d\xb}{dt} &= h^2\si^2 \frac{d\Wb}{d\xb} (\hat N_1(\xb), \hat N_2(\xb))\\
N_i(t)&=\hat N_i(\xb(t)) \mbox{ and } P(t)=\hat P (\xb(t)).
\end{aligned}
\end{equation}
Since $(\hat N_1(\xb), \hat N_2(\xb),\hat P(\xb))$ is piecewise smooth and continuous, the same holds for the right hand side of \eqref{fast}. In particular, the right hand side is Lipschitz and, consequently, solutions exist and are unique. 

Using geometric singular perturbation theory (see, e.g., \cite{hek-10} for a nice review), we prove the following theorem. This theorem provides a sufficient condition for the existence of a stable equilibrium for the eco-evo dynamics and an estimate of the size of its basin of attraction.

\begin{theorem}~\label{thm:fast-slow}  Let $[a,b]$ be a subinterval of $[\theta_1,\theta_2]$ and $\xb^*\in (a,b)$ be a trait value at which $(\hat N_1(\xb), \hat N_2(\xb))$ is continuously differentiable. If $\frac{d\Wb}{d\xb} (\hat N_1(\xb), \hat N_2(\xb),\xb)> 0$ for all $\xb \in [a,\xb^*)$ and $\frac{d\Wb}{d\xb} (\hat N_1(\xb), \hat N_2(\xb))< 0$  for all $\xb\in(\xb^*,b]$,  then for all $\delta>0$ there exists $h_0>0$ such that solutions of \eqref{dyn1}-\eqref{dyn2} satisfy
\[
\lim_{t\to\infty}(N_1(t),N_2(t),P(t),\xb(t))=(\hat N_1(\xb^*), \hat N_2(\xb^*),\hat P(\xb^*),\xb^*)
\]
whenever $N_1(0)N_2(0)P(0)>\delta$, $\xb(0)\in [a,b]$, and $0<h<h_0$. \end{theorem}
 
\paragraph{\textbf{Note}:} We conjecture that $h_0$ in the statement of Theorem~\ref{thm:fast-slow} can be chosen to be independent of $\delta>0$. 

\begin{figure}
\begin{tabular}{cc}
\includegraphics[width=0.5\textwidth]{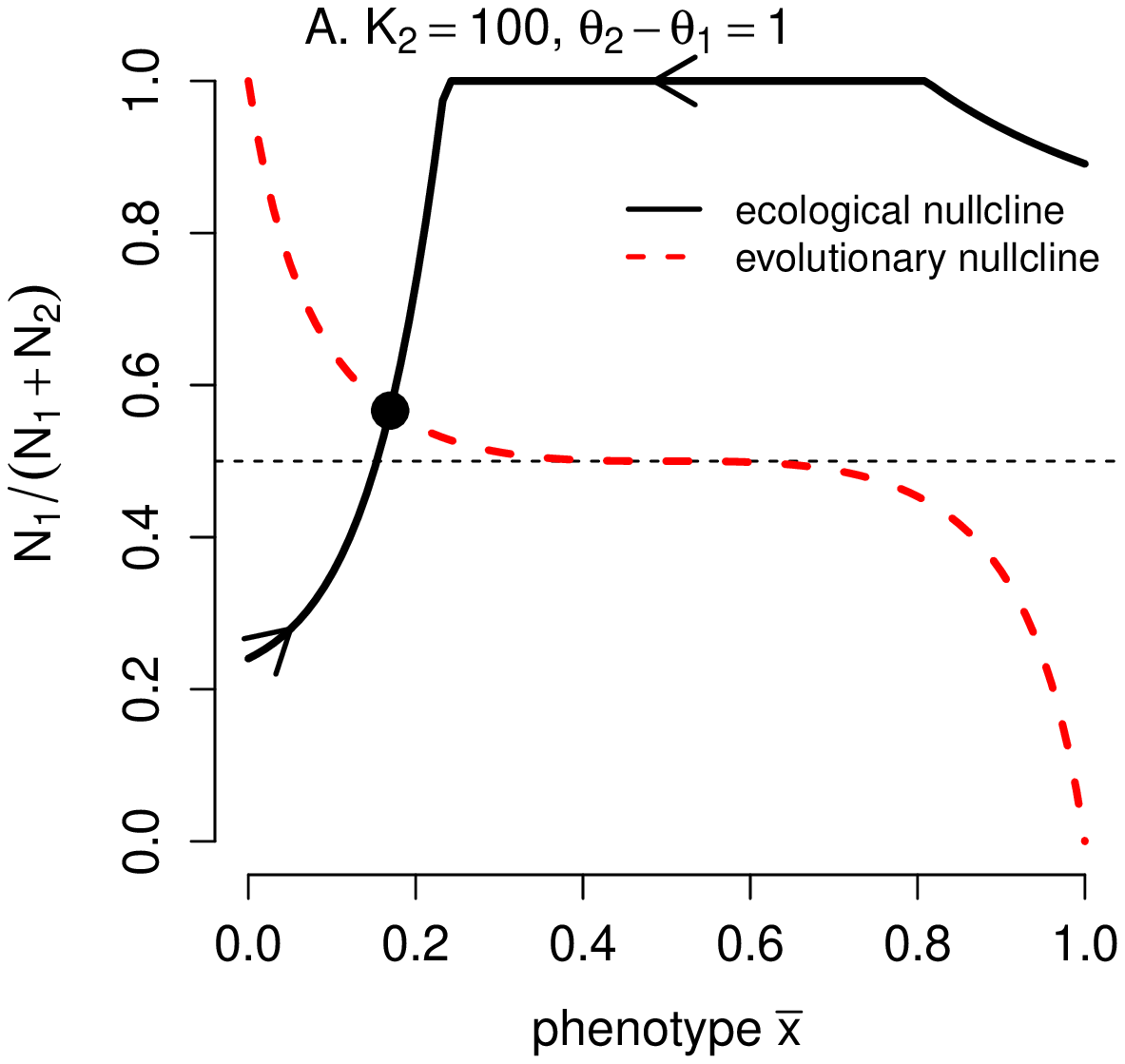}&\includegraphics[width=0.5\textwidth]{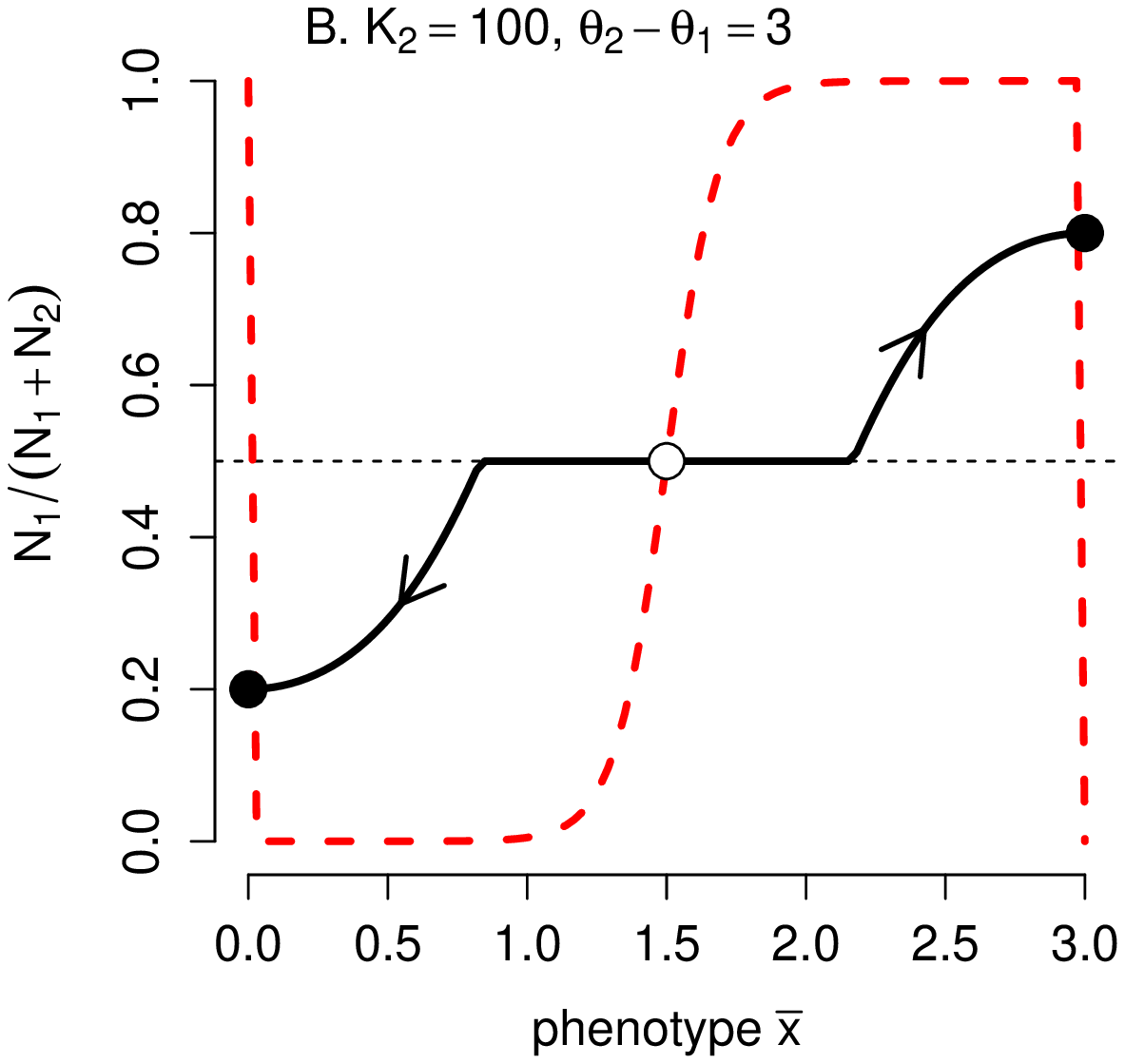}\\
\includegraphics[width=0.5\textwidth]{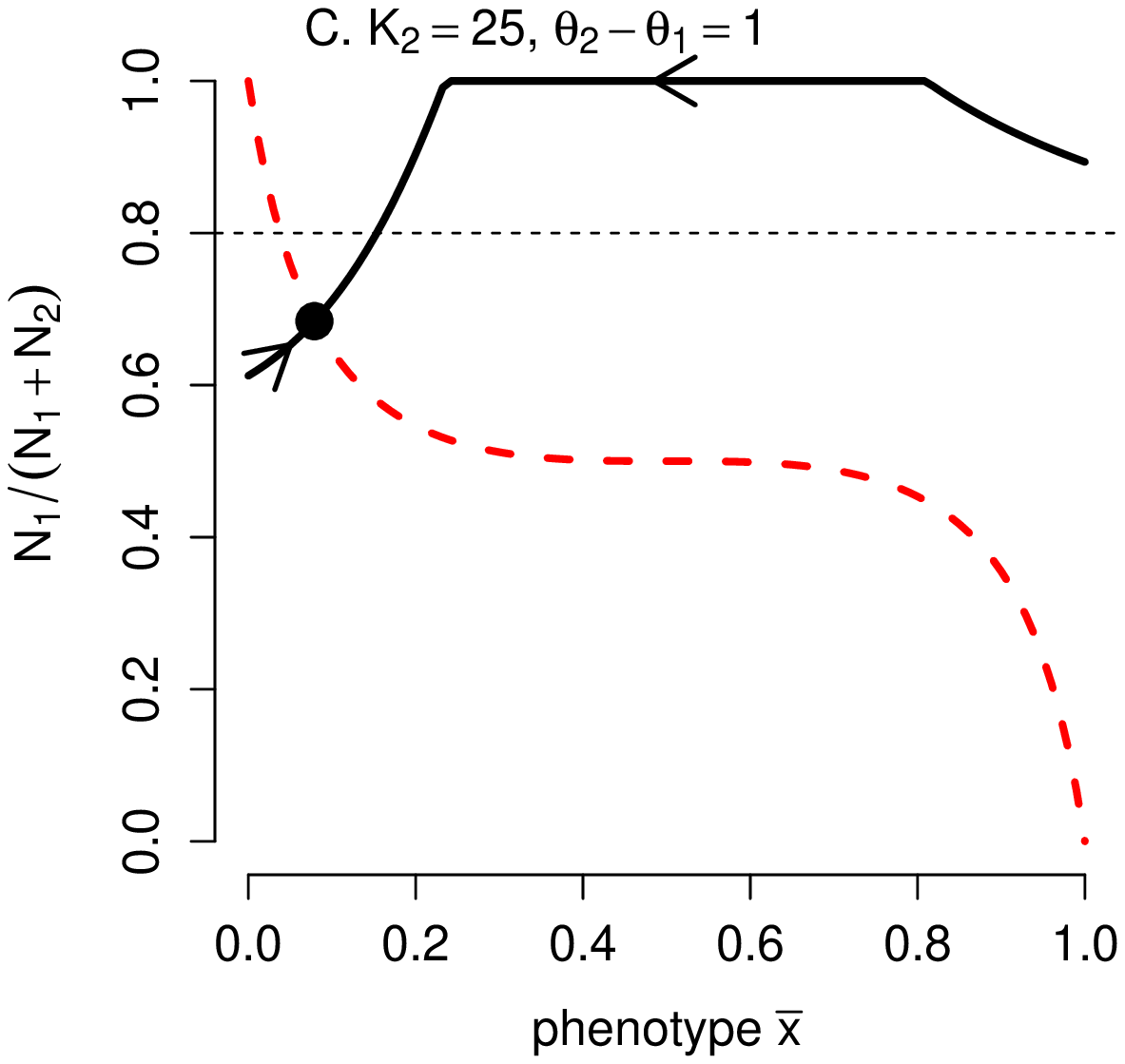}&\includegraphics[width=0.5\textwidth]{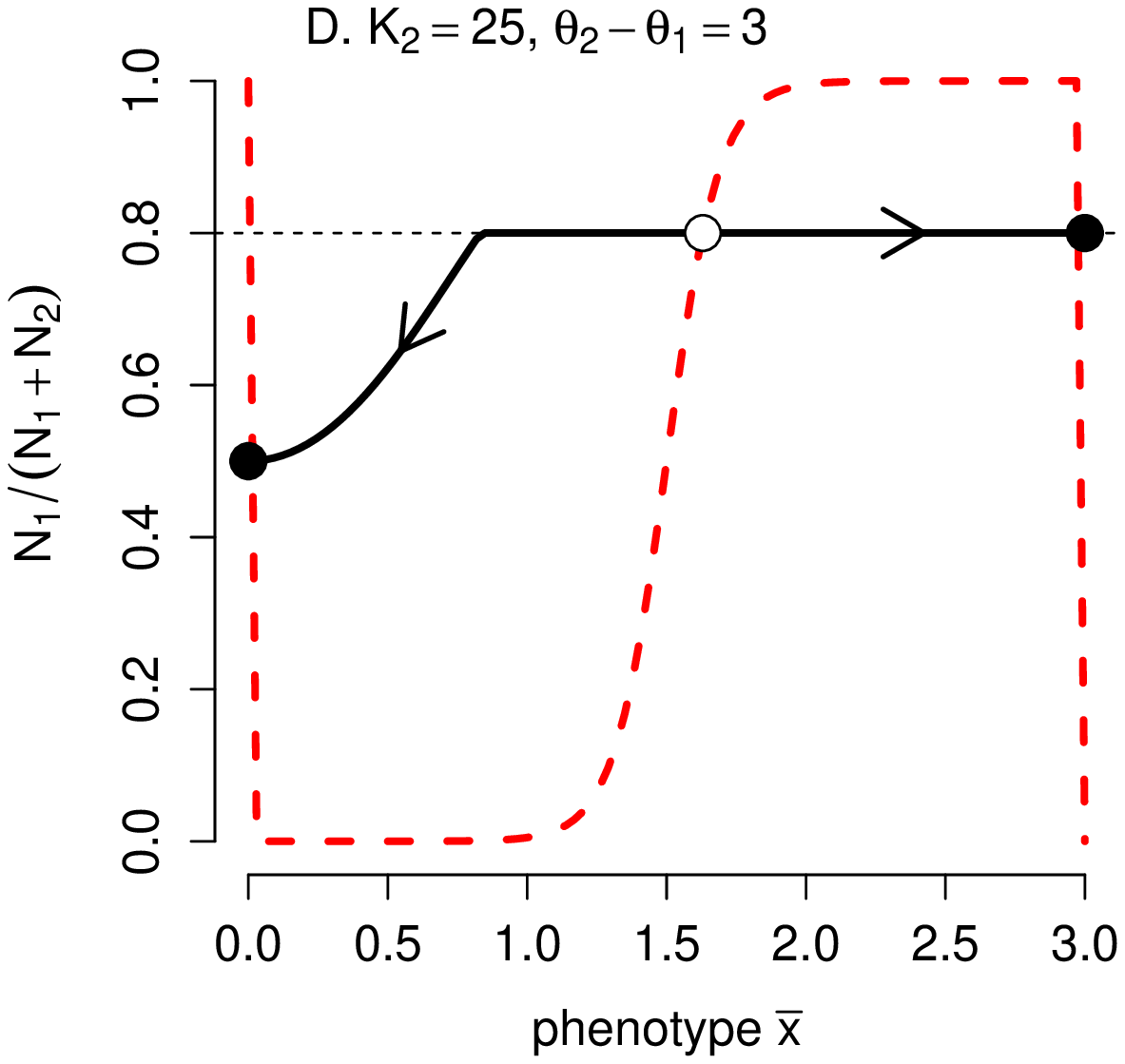}\\
\end{tabular}
\caption{Graphical representation of the eco-evolutionary dynamics for low heritability ($h^2\approx 0$). Black solid lines correspond to the ratio of $\frac{N_1}{N_1+N_2}$ at the ecological quasi-stable states i.e. the ecological nullcline. Dashed red lines correspond to where $\frac{d\xb}{dt}=0$ i.e. the evolutionary nullcline. In the fast slow limit, the dynamics converge toward the ecological nullcline and move to the right or left depending on whether $\frac{d\xb}{dt}>0$ or $\frac{d\xb}{dt}<0$. Eco-evolutionary equilibria occur at the intersection of these nullclines. Theorem~\ref{thm:fast-slow} implies the black equilibria are stable. Parameters: $r_1=1,r_2=0.25$, $K_1=100$, $\alpha_1=\alpha_2=0.01,\tau_1=\tau_2=0.4,\sigma=0.3, e_1=e_2=0.5$, and $d=0.1$. }\label{fig:eco-evo-plane}
\end{figure}

To apply Theorem~\ref{thm:fast-slow}, we use a graphical approach in the plane where the horizontal axis corresponds to the mean trait value and the vertical axis characterizes the quasi-stable-equilibria of the ecological dynamics (Fig.~\ref{fig:eco-evo-plane}).  We use the fraction $\frac{N_1}{N_1+N_2}$ at the quasi-stable-equilibria to characterize the ecological state of the community. This ratio equals zero if prey $1$ is excluded, and equals one if prey $2$ is excluded. When the predator is excluded this ratio equals $\frac{K_1}{K_1+K_2}$. Any other value of this ratio (i.e. not $0$, $K_1/(K_1+K_2)$, or $1$) corresponds to a quasi-stable-equilibria for which all three species coexist. 

The graphical approach continues by drawing two curves in the rectangle $[\theta_1,\theta_2]\times [0,1]$. The first curve, ``the ecological nullcline'', corresponds to the graph of quasi-stable equilibria for the ecological dynamics i.e. the graph of the function: \[\xb\mapsto\frac{\hat N_1(\xb)}{\hat N_1(\xb)+\hat N_2(\xb)}.\]
As the points along this curve correspond to globally stable equilibria for the ecological dynamics, the eco-evo dynamics ``rapidly'' approach and move along these curves in the fast-slow limit.  These ecological nullclines correspond to the solid black curves in Figure~\ref{fig:eco-evo-plane}.  

The second curve, ``the evolutionary nullcline'', is the set of values $(\xb, y)$ such that 
\[
0= y\frac{e_1 \tau_1 \alpha_1 (\theta_1 -\xb)}{(\tau_1^2+\sigma^2)^{3/2}} \exp \left[-\frac{(\xb-\theta_1)^2}{2(\tau_1^2+\sigma^2)}\right]+(1-y)\frac{e_2  \tau_2 \alpha_2 (\theta_2 -\xb)}{(\tau_2^2+\sigma^2)^{3/2}} \exp \left[-\frac{(\xb-\theta_2)^2}{2(\tau_2^2+\sigma^2)}\right]
\]
This second curve corresponds to ratios $y=\frac{N_1}{N_1+N_2}$ at which $\frac{d\xb}{dt}=0$. The mean trait has a negative rate of change ($\frac{d\xb}{dt}<0$) at points lying below the evolutionary nullcline, and a positive rate of change ($\frac{d\xb}{dt}>0$) for points lying above. In Figure~\ref{fig:eco-evo-plane}, the evolutionary nullclines correspond to the dashed red curves.

Intersections between the ecological and evolutionary nullclines correspond to equilibria of \eqref{dyn1}-\eqref{dyn2}. Since the evolutionary nullcline separates the regions in which $\frac{d\xb}{dt}>0$ from the regions in which $\frac{d\xb}{dt}<0$, we can use these graphs to identify eco-evo equilibria satisfying the conditions of Theorem~\ref{thm:fast-slow}.  For instance, in Figure~\ref{fig:eco-evo-plane}, the black points correspond to equilibria for which $\frac{d\xb}{dt}>0$ for lower $\xb$ values along the ecological curve and $\frac{d\xb}{dt}<0$ for higher $\xb$ values along the ecological curve. Hence, Theorem~\ref{thm:fast-slow} implies that these black points correspond to stable equilibria for \eqref{dyn1}-\eqref{dyn2} whenever $h^2>0$ is sufficiently small. Moreover, let $[a,b]$ be an interval of $\xb$ values for which the ecological curve only intersects the evolutionary nullcline at the stable equilibrium. Then for any $\delta>0$, Theorem~\ref{thm:fast-slow} implies that the basin of attraction for these stable equilibria include the set $[\delta,\infty)^3 \times [a,b]$ provided that $h^2>0$ is sufficiently small. 

Figure~\ref{fig:eco-evo-plane} illustrates how changing the carrying capacities of the prey and the strength of the trade-offs influence long-term eco-evolutionary outcomes. When the trade-off is weak, the ecological and evolutionary nullclines intersect at a single point (Fig.~\ref{fig:eco-evo-plane}A,C). In this case, ``most'' initial conditions (i.e. those for which the initial densities are bounded below by an arbitrarily small $\delta>0$) converge to this equilibrium provided $h^2>0$ is sufficiently small. Indeed, we conjecture that this unique intersection point is globally stable i.e. all positive initial conditions converge to these equilibrium. The conjecture following the statement of Theorem~\ref{thm:fast-slow} would imply this result if one could show there is always a unique intersection of the evolutionary and ecological nullclines at weak trade-offs.  

When the strength of the trade-off is strong,  the evolutionary nullcline becomes ``S'' shaped and there are two stable equilibria (Fig.~\ref{fig:eco-evo-plane}B,D). Provided the predator can be sustained on both prey species, the predator can evolve to specialize on one of the prey species  (Fig.~\ref{fig:eco-evo-plane}B). However, if the carrying capacity of one of the prey species is too low and the predator is initially too specialized on this prey species, then the predator evolves to further specialize on this prey and drives itself to extinction despite the possibility of persisting if initially sufficiently specialized on the other prey species (Fig.~\ref{fig:eco-evo-plane}D). 

\begin{figure}
\begin{tabular}{cc}
\includegraphics[width=\textwidth]{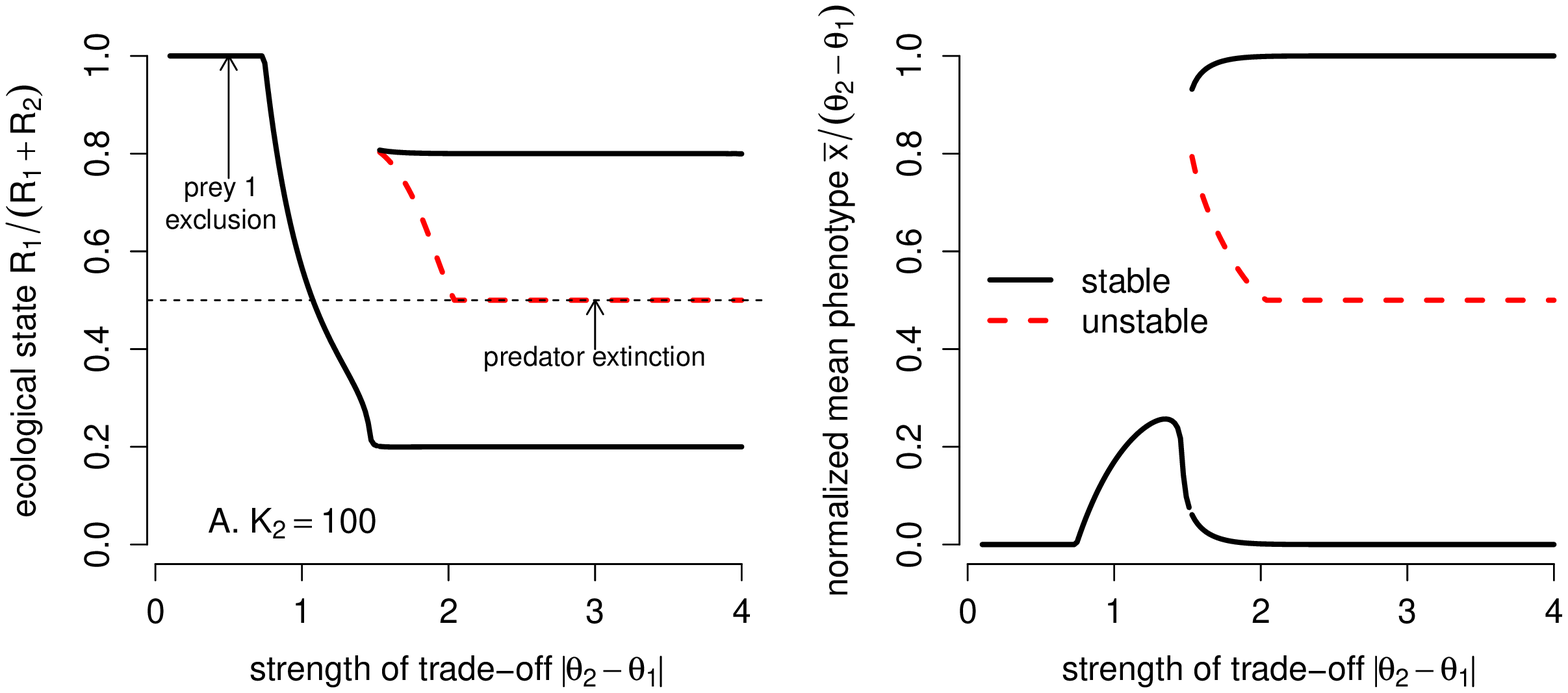}\\
\includegraphics[width=\textwidth]{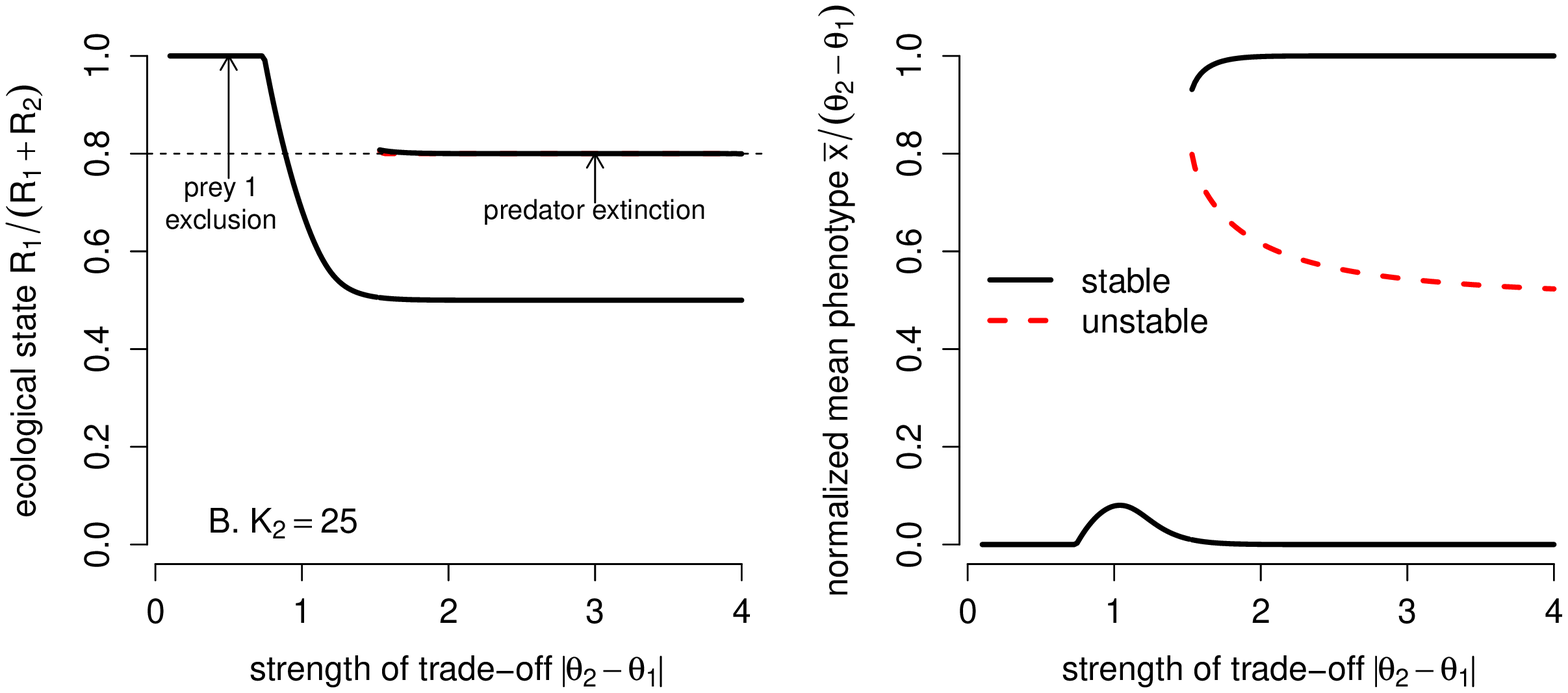}\\
\end{tabular}
\caption{The effect of trade-off strengths of the eco-evolutionary dynamics at equal (A) and unequal (B) prey carrying capacities. Black curves correspond to stable equilibria (for low $h^2$) of the ecological state (left) and the evolutionary state (right). Red dashed curves correspond to unstable (for low $h^2$) equilibria. Parameters: $r_1=1,r_2=0.25,K_1=100,\alpha_1=\alpha_2=0.01,\tau_1=\tau_2=0.4,\sigma=0.3, e_1=e_2=0.5$, and $d=0.1$.}\label{fig:bif}
\end{figure}

Figure~\ref{fig:bif} illustrates how the eco-evolutionary dynamics change as the strength of the trade-off increases. At sufficiently low trade-offs, prey species $2$ is excluded as the superior prey species (species $1$) has a sufficiently large carrying capacity. When the carrying capacities of both prey are sufficiently high, the predator always persists at the stable equilibria, but is excluded at an unstable equilibria when trade-offs are sufficiently high (Figure~\ref{fig:bif}A). At these higher trade-offs, the system is not permanent. However, we conjecture that the system is almost-surely permanent i.e. almost every initial conditions converges to an interior attractor~\citep{jde-04}. Theorem~\ref{thm:fast-slow} implies there are two attractors at which all species coexist. At one attractor, the predator is more specialized on prey $1$ and at the other attractor, it is more specialized on prey $2$. In sharp contrast,  when the carrying capacity of the inferior prey species (species $2$) is low and there are sufficiently high trade-offs, one attractor supports coexistence of all species, while the predator is excluded at the other attractor (Figure~\ref{fig:bif}B).

\section{discussion}

We analyzed a model introduced by \citet{ecology-11b} of two prey species sharing an evolving predator. The trait under selection is continuous and influences the predator's ability to attack the prey species. In particular, there is a trade-off in which traits best-adapted for attacking one prey species aren't best-adapted for attacking the other prey species. We show that the strength of this trade-off determines whether one of the prey species is excluded, the system is permanent in which case coexistence occurs for all initial conditions, or there is contingent coexistence or, more generally, multiple attractors. 
We prove that permanence requires that (i) each prey has a positive per-capita growth rate when the predator is specialized on the other prey and (ii) the predator has a positive per-capita growth rate at all the fitness extrema when the prey are at their carrying capacities.  When top-down effects of the predator are strong (e.g. the predator maximal attack rates or prey carrying capacities are high), we show that the equilibrium excluding the inferior prey (i.e. the prey with the smaller value of $r_i/\ab_i(\theta_i)$) is stable at sufficiently weak trade-offs. Hence, the community may not persist if the predator only experiences weak trade-offs. On the other hand, for sufficiently strong trade-offs, the predator has a negative per-capita growth rate at a fitness minimum when the prey are at their carrying capacities. Consequently, there are initial conditions leading to the exclusion of the predator. In the presence of these strong top-down effects, permanence only occurs at intermediate trade-offs which provide the correct balance of reducing the negative indirect effect of the superior prey on the inferior prey and the predator phenotypes always having a positive per-capita growth rate when the prey are at their carrying capacities. These permanence results highlight how eco-evolutionary feedbacks can alter long-term ecological outcomes. They complement earlier work on eco-evolutionary dynamics of competing species which found eco-evolutionary feedbacks can facilitate or hinder coexistence~\citep{rael-etal-11,vasseur-etal-11}.

When the trait dynamics occur sufficiently slower than the population dynamics, we  characterized the stable equilibria for the eco-evolutionary dynamics. Moreover, we proved that these stable equilibria are nearly globally stable with respect to the ecological state. Namely, provided all species densities are not too low and the predator trait is sufficiently close to its equilibrium value, the eco-evolutionary dynamics converge to the stable equilibrium. Verifying this characterization using a graphical approach, we refined our understanding of the eco-evolutionary dynamics when top-down predator effects are strong. Most notably, at strong trade-offs there are two attractors corresponding to the predator evolving to specialize on one or the other prey species. If the carrying capacities of both prey species are sufficiently high, then all species coexist at these attractors. However, if the carrying capacity of one prey species is low, then one attractor corresponds to extinction of the predator. If the predator population is initially overly specialized on the prey with the lower carrying capacity, the strong trade-off results in selection for continued specialization on this prey species despite it ultimately resulting in extinction.

Our analysis highlights the potential importance of eco-evolutionary feedbacks determining long-term community structure. However, many challenges remain. We used a quantitative genetics framework to model the evolutionary process. While this approach provides a useful first pass on understanding the evolution of continuous traits, genetic variances are likely to change overtime and, consequently, it would be valuable to model the full distributional dynamics to determine the robustness of the conclusions drawn here to this added realism. Comparisons to other models with different genetic architectures, such as clonal evolution or sexual reproduction with only a few loci, will help us understand the role of this architecture  on coexistence. We also assumed that only the predator evolves, but it is likely that the prey would coevolve with the predator. Understanding the resulting eco-evolutionary feedbacks is likely to be very challenging, but ultimating crucial for understanding what the role of coevolution in stabilizing or destabilizing ecological communities. 

\section{Proof of Theorem~\ref{thm:permanence}}
We begin with some definitions and a restatement of the Corollary to Theorem 2 in \cite{garay-89}, which we use in proving the propositions as well as the main proof. Let $E$ be a closed subset of a locally compact metric space $(\mathcal{E}, d)$. Define $d(z,M)=\inf\{d(z,m)|m\in M\}$, for $z\in E$ and $M\subset E$.  Let $\pi$ be a dynamical system on $E$ with $\pi:E\times \R \rightarrow E$ with $\pi(z,0)=z$ and $\pi(\pi(z,t),s)=\pi(z, t+s)$.  
For $z\in E$, $\omega(z)=\cap_{t\geq0} \overline{\{\pi(z,s)|s\in [t,\infty)\}}$ is the $\omega$-limit set and $\alpha(z)=\cap_{t\leq0}\overline{\{\pi(z,s)|s\in (-\infty,s]\}}$ is the $\alpha$-limit set.  The stable set of a compact invariant set $M \subset E$ is $W^+(M)=\{z \in E|\omega(z)\subset M\}$.  A compact invariant set $A$ is an attractor provided there is an open neighborhood $U$ of $A$ such that $\cap_{t\ge 0} \overline{\cup_{s\ge t}\{\pi(z,s)|z\in U\} }=A$.  We say that $\pi$ is dissipative if there is a compact global attractor $S$ such that $W^+(S)=E$. 

A collection $\mathcal{M} =\{M_1, M_2, \dots, M_n\}$ is a \emph{Morse decomposition} for a $\pi$ if $M_1, M_2,...M_n$ are pairwise disjoint, compact, isolated invariant sets with the property that for each $z\in E$ there are integers $i=i(z)$ and $j=j(z)$ with $i\leq j$ and such that $\omega(z)\subset M_i, \alpha(z)\subset M_j$ and if $i=j$, then $z\in M_i=M_j$.
Now we can state the theorem from \cite{garay-89}.

\begin{theorem}\label{thm:garay} 
Let $(\mathcal{E}, d)$ be a locally compact metric space.  Let $\pi$ be a dissipative dynamical system on a closed set $E\subset \mathcal{E}$ with maximal compact invariant set $S$. Let  $\partial E\subset E$ be a closed invariant set and $\mathring{E}=E\setminus \partial{E}$.  If  there exists a Morse decomposition $\mathcal{M} =\{M_1, M_2, \dots, M_n\}$ for $\pi|S\cap \partial{E}$ (the dynamics restricted to $S\cap\partial{E}$) such that for each $i\in\{1, 2, \dots, n\}$
\begin{enumerate}
\item there exists a $\gamma$ such that the set $\{z\in\mathring{E}|d(z,M_i)<\gamma\}$ contains no entire trajectories, and
\item $\mathring{E}\cap W^+(M_i)=\emptyset$
\end{enumerate}
then $\pi$ is permanent, i.e. there exists a $\beta$ such that
\[ 
\liminf_{t \rightarrow \infty} d(\pi(z, t), \partial{E}) \geq \beta
\]
for all $z\in \mathring{E}$. 
\end{theorem}

Note that if $E=\mathbb{R}_+^3\times \mathbb{R}$, with the mean phenotype taking any real values then our system is not dissipative: when both prey are excluded, every point in $\{(0,0,0,\xb) | \xb \in \mathbb{R} \}$ is a fixed point. Thus, to be able to apply Theorem \ref{thm:garay}, we restrict the phenotype space to $[\theta_1, \theta_2]$ and define $E= \mathbb{R}_+^3 \times [\theta_1, \theta_2]$. \eqref{dyn1}-\eqref{dyn2} generates a semi-flow $\pi: \mathbb{R}_+ \times E \rightarrow E$.  Since $\frac{d\xb}{dt}\geq 0$ whenever $\xb=\theta_1$ and $\frac{d\xb}{dt}\leq 0$ whenever $\xb=\theta_2$, $E= \mathbb{R}_+^3 \times [\theta_1, \theta_2]$ forms a forward invariant set.  We define $\partial{E}=\{(N_1, N_2, P, \xb)| N_1N_2P=0\}$ to be the set in which at least one of the species is excluded, and it follows that $\mathring{E}$ is the set in which all three species coexist. 

Dissipativeness of $\pi$ follows from a proof similar to that given in \citep[Lemma 4]{jde-04} that includes the Lotka-Volterra model of apparent competition.

To apply Theorem \ref{thm:garay} to \eqref{dyn1}-\eqref{dyn2}, we begin by proving the two propositions for the two species subsystems. We begin with the proof for the predator-prey subsystem. 

\begin{proof}[Proof of Proposition~\ref{prop:pred-prey}.] We provide the details of the proof for the case $e_1\ab_1(\theta_1)K_1>d$ as the complementary case proceeds similarly. For the duration of this proof, let $E=\R_+\times \{0\} \times \R_+ \times [\theta_1,\theta_2]$ (i.e. the state space restricted to the subsystem without prey $2$) and $\pi$ denote the semi-flow of \eqref{dyn1}-\eqref{dyn2} restricted to $E$. Since $\pi$ is dissipative, there exists a global attractor $S\subset E$. 

We begin by showing that this system is permanent with respect to the boundary set $\partial E = \R_+\times \{0\}\times\{0\}\times [\theta_1,\theta_2]\cup \{0\}\times\R_+\times\{0\}\times [\theta_1,\theta_2]$. To this end, we verify the assumptions of Theorem~\ref{thm:garay}. Define $\mathring E=E\setminus \partial E$. As $\frac{dP}{dt}=-dP$ whenever $N_1=N_2=0$, we have $S\cap \partial E \subset [0,K_1]\times \{0\} \times \{0\} \times [\theta_1,\theta_2]$. A Morse decomposition of $\pi$ restricted to $S\cap \partial E$ is given by $M_1=\{0\}\times \{0\} \times \{0\} \times [\theta_1,\theta_2]$ and $M_2=\{(K_1,0,0,\theta_1)\}$.  To see that $M_1$ is isolated with respect to $\mathring E$ and $W^+(M_1)\cap \mathring E=\emptyset$ (i.e. the conditions of Theorem~\ref{thm:garay}), choose $\epsilon>0$ sufficiently small that $r_1N_1(1-N_1/K_1)-\ab_1(\theta_1)N_1P\ge r_1 N_1/2$ whenever $N_1\le \epsilon$, $P\le \epsilon$. Suppose to the contrary that $M_1$ is not isolated from $\mathring E$. Then there exists an invariant set $B$ in $\mathring E$ such that $\max_{(N_1,0,P,\xb)\in B} |N_1|+|P|<\epsilon$. For any solution $(N_1(t),0,P(t),\xb(t))$ with initial condition $(N_1(0),0,P(0),\xb(0))\in B$ satisfying $N_1(0)P(0)>0$, we have that $N_1'(t)\ge r_1 N_1(t)/2$ for all $t\ge 0$ which implies $\lim_{t\to\infty} N_1(t)=\infty$. However, this violates the fact that $B$ is bounded. Using a similar argument, there can be no initial condition in $\mathring E$ whose $\omega$-limit set lies in $M_1$. Hence $W^+(M_1)\cap \mathring E=\emptyset$.  

As $e_1\ab_1(\theta_1)K_1>d$, we can find a neighborhood $U$ of $M_2$ and $\delta>0$ such that $e_1 \ab_1(\xb) N_1 -d \ge \delta$ whenever $(N_1,0,P,\xb) \in U$. Hence, using a similar argument as for the set $M_1$, we can conclude that $M_2$ is isolated relative to $\mathring E$ and $W^+(M_2)\cap \mathring E =\emptyset$. Theorem~\ref{thm:garay} implies $\pi$ is permanent with respect to $\partial E$. 

Now consider a solution $N_1(t),P(t),\xb(t)$ such that $N_1(0)P(0)>0$. Since the system is permanent, there exists $\beta>0$ such that $\liminf_{t\to\infty}\min\{N_1(t),P(t)\}\ge \beta $.  Provided $N_1>0$, the function
\[
\frac{\partial \Wb}{\partial \xb} = \frac{e_1 N_1 \tau_1 \alpha_1 (\theta_1 -\xb)}{(\tau_1^2+\sigma^2)^{3/2}} \exp \left[-\frac{(\xb-\theta_1)^2}{2(\tau_1^2+\sigma^2)}\right]
\]
has a unique zero at $\xb=\theta_1$, is positive for $\xb<\theta_1$, and is negative for $\xb>\theta_1$. Since $N_1(t)\ge \beta/2$ for $t$ sufficiently large, it follows that $\lim_{t\to\infty}\xb(t)=\theta_1$. As the ecological dynamics \eqref{dyn1} with the fixed value of $\xb=\theta_1$ has a global stable equilibria positive equilibrium $(d/(e_1\ab_1(\theta_1)), r_1(1-\hat{N}_1/K_1)/\ab_1(\theta_1), \theta_1)$, \citet[Theorem 1.8]{mischaikow-etal-95} implies the $\omega$-limit set of $(N_1(t),P(t),\xb(t))$ is this equilibrium. \end{proof}

Next we prove the proposition for the two prey subsystem. 

\begin{proof}[Proof of Proposition~\ref{prop:prey-prey}.] 
Assume $(N_1(t), N_2(t), 0, \xb(t))$ is a solution to \eqref{dyn1}-\eqref{dyn2} with $N_1(0)N_2(0)>0$.  Prey species $i$ experiences logistic growth, so $\lim_{t\rightarrow \infty} N_i(t) = K_i$ for $i=1, 2$. Let $A=\{(K_1, K_2, \xb) | \xb\in \mathbb{R}\}$, which is invariant, since $\frac{dN_i}{dt}=0$ for $i=1, 2$. Let $f(\xb)=\sigma_G^2 \frac{d\Wb}{d\xb}(K_1, K_2, \xb)$ be the phenotype dynamics when $N_i=K_i$ for $i=1, 2$.  
\citet[Theorem 1.8]{mischaikow-etal-95} implies that $\omega((N_1(0), N_2(0), 0, \xb(0)))\subset \{(K_1, K_2, 0, \xb)| f(\xb)=0\}$. Since $f$ is an analytic function and $f(\xb)<0$ for $\xb$ sufficiently negative or positive, the zeros of $f$ are isolated.  Since an $\omega$-limit is a connected set, $\omega((N_1(0), N_2(0), 0, \xb(0)))$ is a single point in $\{(K_1, K_2, 0, \xb)| f(\xb)=0\}$.
\end{proof}

Now, we prove the main theorem.  Without loss of generality, assume $\ab_1(\theta_1)e_1K_1 \geq \ab_2(\theta_2)e_2K_2$.  Then, there are three cases to consider: (i) $d\geq \ab_1(\theta_1)e_1K_1$, (ii) $\ab_1(\theta_1)e_1K_1 >d\geq  \ab_2(\theta_2)e_2K_2$ or (iii) $\ab_2(\theta_2)e_2K_2>d$.  

Recall that $E=\mathbb{R}_+^3 \times [\theta_1, \theta_2]$. Since $\pi$ is dissipative in $E$, there is a global attractor $S\subset E$. For each of the three cases, we claim there is  a Morse decomposition for $\pi|S\cap\partial{E}$ that satisfies the assumptions of Theorem \ref{thm:garay}. 

We start with case (iii). Let $\mathcal{M} = \{M_6, M_5, M_4, M_3, M_2, M_1\}$, where $M_6= \{(0,0,0)\} \times [\theta_1, \theta_2], M_5=\{(K_1, 0, 0, \theta_1)\}, M_4=\{(0, K_2, 0, \theta_2)\}, M_3=\{(\hat{N}_1, 0, \hat{P}_1, \theta_1)\}, $ $M_2=\{(0, \hat{N}_2, \hat{P}_2, \theta_2)\}$ where $\hat{N}_i=\frac{d}{e_i\ab_i(\theta_i)}$ and $\hat{P}_i=\frac{r_i(1-\frac{\hat{N}_i}{K_i})}{\ab_i(\theta_i)}$, and $M_1= \{(K_1, K_2, 0)\}\times [\xb_1, \xb_2]$ where $\xb_1=\min_{\xb\in Q}\xb$ and $\xb_2= \max_{\xb\in Q}\xb$. 

We claim that $\mathcal{M}$ is a Morse decomposition for $S\cap \partial E$. We need to verify the $\omega$-limit and $\alpha$-limit sets property of Morse decompositions.  Let $z\in (S\cap \partial{E})\setminus \mathcal{M}$, where $z=(N_1, N_2, P, \xb)$. Then, one of the following holds:

\begin{enumerate}[label=(\alph*)]
\item $N_2=P=0, N_1>0$
\item $N_2=0, N_1>0, P>0$
\item $N_1=P=0, N_2>0$
\item $N_1=0, N_2>0, P>0$
\item $P=0, N_1>0, N_2>0$
\end{enumerate}

If (a) holds, then $\omega(z)= M_5$ and $\alpha(z)\subset M_6$.  If (b) holds, then Proposition~\ref{prop:pred-prey} implies that  $\omega(z)= M_3$.  
The invariance of $\alpha$-limits and Proposition \ref{prop:pred-prey} imply $\alpha(z)\subset \{(N_1, N_2, P, \xb) | N_1P=0, N_2=0\}\cup M_3$.  
Then, by \citet[Proposition 1.5]{mischaikow-etal-95}, $\alpha(z)\subset M_6$, $\alpha(z)= M_5$, or $\alpha(z)=M_3$.  Since $(\hat{N}_1, \hat{P}_1)$ is globally stable for $\pi|\{(N_1,0,P,\theta_1)|N_1,P\in\R\}$, and $\frac{d\xb}{dt}\leq 0$ whenever $N_2=0$ and $\xb\geq \theta_1$, $\omega(z)=\alpha(z)=M_3$ implies that $z\in M_3$.  Cases (c) and (d) follow similarly. If (e) holds, then $\omega(z) \subset M_1$, by Proposition 2, and either $\alpha(z) \subset M_5$, $\alpha(z) \subset M_4$, or $\alpha(z) \subset M_6$.  

Thus, we have shown that $\mathcal{M}_3$ forms a Morse decomposition for $\pi| S\cap \partial{E}$ for case (iii).

Finally, we verify the two assumptions of Theorem \ref{thm:garay} using arguments similar to those made in the proof of Proposition \ref{prop:pred-prey}. To show that $M_1$ is isolated for $\mathring E$ and $W^+(M_1)\cap \mathring E=\emptyset$, recall by assumption that  $e_1\ab_1(\xb^*)K_1+e_2\ab_2(\xb^*)K_2>d$ for all $\xb^*\in Q$.   Hence, there exists a neighborhood $U$ of $M_1$ and $\delta>0$ such that $e_1\ab_1(\xb)N_1+e_2\ab_2(\xb)N_2-d>\delta$ for all $(N_1,N_2,P,\xb)\in U$. Now suppose to the contrary that $M_1$ is not isolated from $\mathring E$ or $W^+(M_1)\cap \mathring E \neq \emptyset$. Then there exists a solution $(N_1(t),N_2(t),P(t),\xb(t))$ which lies in $U\cap \mathring E$ for all $t\ge 0$. But this implies that $P'(t)\ge \delta P(t)$ for all $t\ge 0$ and, consequently, $\lim_{t\to\infty}P(t)=\infty$, a contradiction. We can use the same type of argument to show that $M_i$ is isolated from $\mathring E$ and $W^+(M_i)\cap \mathring E=\emptyset$ for $i=2,\dots,6$. Specifically, for $M_6$, we use a neighborhood $U$ and $\delta>0$ such that $r_1(1-N_1/K_1)-\ab_1(\xb)P>\delta$ for all $(N_1,N_2,P,\xb)\in U$. For $M_5$ and $M_4$, we use a neighborhood $U$ of $M_5\cup M_4$ and $\delta>0$ such that $e_1\ab_1(\xb)N_1+e_2\ab_2(\xb)N_2-d>\delta$ for all $(N_1,N_2,P,\xb)\in U$.  For $M_3$, we use  a neighborhood  $U$ of $M_3$ and $\delta>0$ such that $r_2(1-N_2/K_2)-\ab_2(\xb)P>\delta$ for all $(N_1,N_2,P,\xb)\in U$. Finally, for $M_2$, we use a neighborhood $U$ of $M_2$ and $\delta>0$ such that $r_1(1-N_1/K_1)-\ab_1(\xb)P>\delta$ for all $(N_1,N_2,P,\xb)\in U$. It follows that by Theorem \ref{thm:garay} $\pi|E$ is permanent for case (iii).

We can provide permanence for cases (i) and (ii) in a similar manner using slightly different Morse decompositions. Specifically, for case (i), we can use the Morse decomposition given by  $\mathcal{M} = \{M_4, M_3, M_2, M_1\}$, where $M_4= (0,0,0) \times [\theta_1, \theta_2], M_3=(K_1,0,0, \theta_1), M_2= (0,K_2,0, \theta_2)$, and $M_1=[\xb_1, \xb_2] \times (K_1, K_2, 0)$. For case (ii), we can use the Morse decomposition given by  $\mathcal{M} = \{M_5, M_4, M_3, M_2, M_1\}$, where $M_5=[\theta_1, \theta_2] \times (0,0,0), M_4=(K_1, 0, 0, \theta_1), M_3=(0, K_2, 0, \theta_2), M_2=(\hat{N}_1, 0, \hat{P}_1, \theta_1)$, and $M_1=[\xb_1, \xb_2] \times (K_1, K_2, 0)$. 

Next, we extend our proof to any initial phenotype in $\xb(0) \in \mathbb{R}$ and $N_1(0)N_2(0)P(0)>0$.  First note that there exists a $\delta$ such that $\limsup_{t\rightarrow \infty} \max\{N_1(t), N_2(t)\}>\delta$. It follows that there exists a $\gamma$ such that whenever $\xb(t)<\theta_1$, there is a range $[t+s_1, t+s_2]$ with $s_1>0$ and $s_2>0$ such that $\frac{d\xb}{dt}>\gamma$. Thus, there exists a $T_1$, such that $\xb(T_1)\geq \theta_1$.  Similarly, we can show there exists $T_2$ such that $\xb(T_2)\leq \theta_2$. 

Finally, if inequality (a) in the statement of Theorem~\ref{thm:permanence} is reversed for $i=2$, then the Jacobian matrix at \[
(N_1,N_2,P,\xb)=\left(\frac{d}{e_1\ab_1(\theta_1)},0, \frac{r_1(1-\frac{\hat{N}_1}{K_1})}{\ab_1(\theta_1)}, \theta_1\right)\]
has the following sign structure in the $N_1,P,\xb,N_2$ coordinate system
\[
\begin{pmatrix}
-&-&*&*\\
+&0&*&*\\
0&0&-&*\\
0&0&0&-
\end{pmatrix}
\]
where $*$ indicates a term with a positive, negative or zero sign. This matrix has an upper triangular block structure with $2\times 2$, $1\times 1$, and $1\times 1$ blocks down the diagonal. Each of these diagonal blocks has a sign structure that implies the eigenvalues of these blocks have negative real parts. Hence, this equilibrium is stable by the stable manifold theorem, see e.g. \citep[Theorem 1.3.2]{guckenheimer-holmes-83}. Alternatively, the partial converse when inequality (b) in the statement of Theorem~\ref{thm:permanence} is reversed follows immediately from the center manifold theorem, see e.g. \citep[Theorem 3.2.1]{guckenheimer-holmes-83}.  

\section{Proof of Theorem~\ref{thm:fast-slow}}
Define $g(\xb)=\frac{d\Wb}{d\xb} (\hat N_1(\xb), \hat N_2(\xb),\xb)$. For any set $I\subset [\theta_1,\theta_2]$, define $\mathcal{E}(I)=\{(\hat N_1(\xb), \hat N_2(\xb), \hat P (\xb),\xb)|\xb \in I\}$. Let $\eta>0$ be such that $J=[\xb^*-\eta,\xb^*+\eta]\subset (a,b)$ and the graph of $\mathcal{E}(J)$ is continuously differentiable.  

Choose a compact neighborhood $U$ of $\mathcal{E}([a,b])$ such that (i) there exists a $h_1>0$ such that the maximal invariant set $I(h)$ in $V=U \cap \R^3_+ \times J$ for \eqref{dyn1}-\eqref{dyn2} with $0<h<h_1$ is contained in a forward invariant set that is homeomorphic to a closed interval,  (ii) $U$ is forward invariant and contains the global attractor for the dynamics of \eqref{dyn1}-\eqref{dyn2} for $h=0$ restricted to $(0,\infty)^3\times [a,b]$, (iii) $\frac{d\Wb}{d\xb} (N_1, N_2,\xb)> 0$ for all $(N_1,N_2,P,\xb)\in U\cap \R^3 \times [a,\xb^-\eta]$ and $\frac{d\Wb}{d\xb} (N_1, N_2,\xb)< 0$ for all $(\xb,N_1,N_2,P)\in U\cap \R^3 \times [\xb^+\eta,b]$, and (iv) $(\hat N_1(\xb^*),\hat N_2(\xb^2),\hat P(\xb^*),\xb^*)$ is the only equilibrium in $V$ for \eqref{dyn1}-\eqref{dyn2} for $h>0$. A neighborhood satisfying the first condition follows from geometric singular perturbation theory~\citep[Theorem 1]{fenichel-71}. A possibly smaller neighborhood satisfying the second condition follows from $\mathcal{E}([a,b])$ being a global attractor for the dynamics of $h=0$ restricted to $[a,b]\times\R_+^3$. A possibly even smaller neighborhood satisfying the third condition follows from our assumption that $g(\xb)>0$ for $\xb\in[a,\xb^*)$ and $g(\xb^*)<0$ for $\xb\in (\xb^*,b]$. Finally, a possibly even smaller neighborhood satisfying the fourth condition follows from the equilibria of \eqref{dyn1}-\eqref{dyn2} being independent of $h>0$, and our assumption that $(\hat N_2(\xb),\hat N_1(\xb), \hat P(\xb)$ being stable for the ecological dynamics and $\frac{d\Wb}{d\xb}(\hat N_1(\xb),N_2(\xb),\hat P(\xb))\neq 0$ in $[a,b]$ except at $\xb^*$.

Let $C=[\epsilon/2,M]^3 \times [a,b]$ (here $M$ is chosen so that $[0,M]^3\times [\theta_1,\theta_2]$ contains the global attractor for the dynamics.)  For $h=0$, compactness of $K$, condition (ii), and continuous dependence of solutions on initial conditions implies that there is a $T>0$ so that any solution starting in $K$ enters $U$ by time $T$. Continuity of the solutions of \eqref{dyn1}-\eqref{dyn2} with respect to the initial conditions and parameters, and compactness of $C$ implies that this holds for any $h\in [0,h_2]$ for $h_2>0$ sufficiently small. 

Let $h_0=\min\{h_1,h_2\}$ and choose $h\in (0,h_0)$. For any initial condition in $K$, $(N_1(t),N_2(t),P(t),\xb(t))\in U$ for all $t\ge T$. As $\frac{d\xb}{dt}>0$ in $U\cap \R^3 \times [a,\xb^*-\eta]$ and $\frac{d\xb}{dt}<0$ in $U\cap \R^3\times [\xb^*+\eta,b]$, this solution enters and remains in $V$ for all $t$ sufficiently large. In particular the $\omega$-limit set of the solution must lie in $V$. As the maximal invariant set in $V$ is contained in a forward invariant set homeomorphic to a closed interval and the only equilibrium in $V$ is $(\hat N_1(\xb^*),\hat N_2(\xb^*),\hat P(\xb^*),\xb^*)$ (see condition (iv) above), the maximal invariant set is this equilibrium. In particular, the $\omega$-limit set of the solution is this equilibrium and the proof is complete. 


\begin{thebibliography}{42}
\providecommand{\natexlab}[1]{#1}
\providecommand{\url}[1]{\texttt{#1}}
\expandafter\ifx\csname urlstyle\endcsname\relax
  \providecommand{\doi}[1]{doi: #1}\else
  \providecommand{\doi}{doi: \begingroup \urlstyle{rm}\Url}\fi

\bibitem[Abrams(2006{\natexlab{a}})]{abrams-06a}
P.A. Abrams.
\newblock The effects of switching behavior on the evolutionary diversification
  of generalist consumers.
\newblock \emph{American Naturalist}, 168:\penalty0 645--659,
  2006{\natexlab{a}}.

\bibitem[Abrams(2006{\natexlab{b}})]{abrams-06b}
P.A. Abrams.
\newblock Adaptive change in the resource-exploitation traits of a generalist
  consumer: the coevolution and coexistence of generalists and specialists.
\newblock \emph{Evolution}, 60:\penalty0 427--439, 2006{\natexlab{b}}.

\bibitem[Ara{\'u}jo et~al.(2008)Ara{\'u}jo, Guimar\~{a}es Jr, Svanb{\"a}ck,
  Pinheiro, Guimar{\~a}es, Reis, and Bolnick]{araujo-etal-08}
M.S. Ara{\'u}jo, P.R. Guimar\~{a}es Jr, R.~Svanb{\"a}ck, A.~Pinheiro,
  P.~Guimar{\~a}es, S.F. Reis, and D.I. Bolnick.
\newblock Network analysis reveals contrasting effects of intraspecific
  competition on individual vs. population diets.
\newblock \emph{Ecology}, 89:\penalty0 1981--1993, 2008.

\bibitem[Brose(2010)]{brose-10}
U.~Brose.
\newblock Body-mass constraints on foraging behaviour determine population and
  food-web dynamics.
\newblock \emph{Functional Ecology}, 24:\penalty0 28--34, 2010.

\bibitem[Brose et~al.(2006)Brose, Jonsson, Berlow, Warren, Banasek-Richter,
  Bersier, Blanchard, Brey, Carpenter, Blandenier, et~al.]{brose-etal-06}
U.~Brose, T.~Jonsson, E.~L. Berlow, P.~Warren, C.~Banasek-Richter, L.~Bersier,
  J.~L. Blanchard, T.~Brey, S.~R. Carpenter, M.~Blandenier, et~al.
\newblock Consumer-resource body-size relationships in natural food webs.
\newblock \emph{Ecology}, 87:\penalty0 2411--2417, 2006.

\bibitem[B\"urger and Gimelfarb(2004)]{burger-gimelfarb-04}
R.~B\"urger and S.~Gimelfarb.
\newblock The effects of intraspecific competition and stabilizing selection on
  a polygenic trait.
\newblock \emph{Genetics}, 167:\penalty0 1425--1443, 2004.

\bibitem[Cobb et~al.(2010)Cobb, Meentemeyer, and Rizzo]{cobb-etal-10}
R.C. Cobb, R.K. Meentemeyer, and D.M. Rizzo.
\newblock Apparent competition in canopy trees determined by pathogen
  transmission rather than susceptibility.
\newblock \emph{Ecology}, 91:\penalty0 327--333, 2010.

\bibitem[Fenichel(1971)]{fenichel-71}
N.~Fenichel.
\newblock Persistence and smoothness of invariant manifolds for flows.
\newblock \emph{Indiana Univ. Math. J.}, 21:\penalty0 193--226, 1971.

\bibitem[Fussmann et~al.(2007)Fussmann, Loreau, and Abrams]{fussmann-etal-07}
G.~F. Fussmann, M.~Loreau, and P.~A. Abrams.
\newblock {Eco-evolutionary dynamics of communities and ecosystems}.
\newblock \emph{Functional Ecology}, 21:\penalty0 465--477, 2007.

\bibitem[Garay(1989)]{garay-89}
B.~M. Garay.
\newblock Uniform persistence and chain recurrence.
\newblock \emph{Journal of Mathematical Analysis and Applications},
  139:\penalty0 372--382, 1989.

\bibitem[Guckenheimer and Holmes(1983)]{guckenheimer-holmes-83}
J.~Guckenheimer and P.~Holmes.
\newblock \emph{Nonlinear oscillations, dynamical systems, and bifurcations of
  vector fields}, volume~42.
\newblock Springer Science \& Business Media, 1983.

\bibitem[Hek(2010)]{hek-10}
G.~Hek.
\newblock Geometric singular perturbation theory in biological practice.
\newblock \emph{Journal of Mathematical Biology}, 60:\penalty0 347--386, 2010.

\bibitem[Holt and Lawton(1993)]{holt-lawton-93}
R.~D. Holt and J.~H. Lawton.
\newblock Apparent competition and enemy-free space in insect host-parasitoid
  communities.
\newblock \emph{American Naturalist}, 142:\penalty0 623--645, 1993.

\bibitem[Holt and Lawton(1994)]{holt-lawton-94}
R.~D. Holt and J.~H. Lawton.
\newblock The ecological consequences of shared natural enemies.
\newblock \emph{Annual Review of Ecology Evolution and Systematics},
  25:\penalty0 495--520, 1994.

\bibitem[Holt(1977)]{holt-77}
R.D. Holt.
\newblock Predation, apparent competition and the structure of prey
  communities.
\newblock \emph{Theoretical Population Biology}, 12:\penalty0 197--229, 1977.

\bibitem[Hutson and Schmitt(1992)]{hutson-schmitt-92}
V.~Hutson and K.~Schmitt.
\newblock Permanence and the dynamics of biological systems.
\newblock \emph{Mathematical Biosciences}, 111:\penalty0 1--71, 1992.

\bibitem[Lande(1976)]{lande-76}
R.~Lande.
\newblock Natural selection and random genetic drift in phenotypic evolution.
\newblock \emph{Evolution}, 30:\penalty0 314--334, 1976.

\bibitem[MacArthur(1955)]{macarthur-55}
R.~MacArthur.
\newblock Fluctuations of animal populations and a measure of community
  stability.
\newblock \emph{Ecology}, 36:\penalty0 533--536, 1955.

\bibitem[Matthews et~al.(2010)Matthews, Marchinko, Bolnick, and
  Mazumder]{matthews-etal-10}
B.~Matthews, K.B. Marchinko, D.I. Bolnick, and A.~Mazumder.
\newblock Specialization of trophic position and habitat use by sticklebacks in
  an adaptive radiation.
\newblock \emph{Ecology}, 91:\penalty0 1025--1034, 2010.

\bibitem[Mischaikow et~al.(1995)Mischaikow, Smith, and
  Thieme]{mischaikow-etal-95}
K.~Mischaikow, H.~Smith, and H.~R. Thieme.
\newblock Asymptotically autonomous semiflows: Chain recurrence and {L}yapunov
  functions.
\newblock \emph{Transactions of the American Mathematical Society},
  347:\penalty0 1669--1685, 1995.

\bibitem[Morris et~al.(2001)Morris, M\"{u}ller, and Godfray]{morris-etal-01}
R.~J. Morris, C.~B. M\"{u}ller, and H.C.~J. Godfray.
\newblock Field experiments testing for apparent competition between primary
  parasitoids mediated by secondary parasitoids.
\newblock \emph{Journal of Animal Ecology}, 70:\penalty0 301--309, 2001.

\bibitem[Mueller and Godfray(1997)]{mueller-godfray-97}
C.~B. Mueller and H.~C.~J. Godfray.
\newblock Apparent competition between two aphid species.
\newblock \emph{Journal of Animal Ecology}, 66:\penalty0 57--64, 1997.

\bibitem[Nurmi and Parvinen(2013)]{nurmi-parvinen-13}
T.~Nurmi and K.~Parvinen.
\newblock Evolution of specialization under non-equilibrium population
  dynamics.
\newblock \emph{Journal of Theoretical Biology}, 321:\penalty0 63--77, 2013.

\bibitem[Patel and Schreiber(In review)]{patel-schreiber-preprint}
S.~Patel and S.J. Schreiber.
\newblock Evolutionary driven regime shifts in ecological systems with
  intraguild predation.
\newblock In review.

\bibitem[Pelletier et~al.(2009)Pelletier, Garant, and
  Hendry]{pelletier-etal-09}
F.~Pelletier, D.~Garant, and A.P. Hendry.
\newblock {Eco-evolutionary dynamics.}
\newblock \emph{Philosophical Transactions of the Royal Society of London.
  Series B, Biological sciences}, 364:\penalty0 1483--1489, 2009.
\newblock ISSN 1471-2970.

\bibitem[Rael et~al.(2011)Rael, Vincent, and Cushing]{rael-etal-11}
R.C. Rael, T.L. Vincent, and J.M. Cushing.
\newblock Competitive outcomes changed by evolution.
\newblock \emph{Journal of Biological Dynamics}, 5:\penalty0 227--252, 2011.

\bibitem[Rand(2003)]{rand-03}
T.~A. Rand.
\newblock Herbivore-mediated apparent competition between two salt marsh forbs.
\newblock \emph{Ecology}, 84:\penalty0 1517--1526, 2003.

\bibitem[Rand et~al.(2004)Rand, Russell, and Louda]{rand-etal-04}
T.~A. Rand, F.~L. Russell, and S.~M. Louda.
\newblock Local- vs. landscape-scale indirect effects of an invasive weed on
  native plants.
\newblock \emph{Weed Technology}, 18:\penalty0 1250--1254, 2004.

\bibitem[Robinson(2000)]{robinson-00}
B.W. Robinson.
\newblock Trade offs in habitat-specific foraging efficiency and the nascent
  adaptive divergence of sticklebacks in lakes.
\newblock \emph{Behaviour}, 137:\penalty0 865--888, 2000.

\bibitem[Rott and Godfray(1998)]{rott-godfray-98}
M.~Rott and H.C.J. Godfray.
\newblock Indirect population interaction between two aphid species.
\newblock \emph{Ecology Letters}, 1:\penalty0 99--103, 1998.

\bibitem[Rueffler et~al.(2006)Rueffler, Van~Dooren, and Metz]{rueffler-etal-06}
C.~Rueffler, T.J.M. Van~Dooren, and J.A.J. Metz.
\newblock The evolution of resource specialization through frequency-dependent
  and frequency-independent mechanisms.
\newblock \emph{American Naturalist}, 167:\penalty0 81--93, 2006.

\bibitem[Schoener(2011)]{schoener-11}
T.~W. Schoener.
\newblock {The newest synthesis: understanding the interplay of evolutionary
  and ecological dynamics.}
\newblock \emph{Science}, 331:\penalty0 426--429, 2011.

\bibitem[Schreiber(2000)]{jde-00}
S.~J. Schreiber.
\newblock Criteria for ${C}^r$ robust permanence.
\newblock \emph{Journal of Differential Equations}, 162:\penalty0 400--426,
  2000.

\bibitem[Schreiber(2004)]{jde-04}
S.~J. Schreiber.
\newblock Coexistence for species sharing a predator.
\newblock \emph{Journal of Differential Equations}, 196:\penalty0 209--225,
  2004.

\bibitem[Schreiber and Tobiason(2003)]{jmb-03}
S.~J. Schreiber and G.~A. Tobiason.
\newblock The evolution of resource use.
\newblock \emph{Journal of Mathematical Biology}, 47:\penalty0 56--78, 2003.

\bibitem[Schreiber et~al.(2011)Schreiber, Bolnick, and B{\"u}rger]{ecology-11b}
S.~J. Schreiber, D.~Bolnick, and R.~B{\"u}rger.
\newblock The community effects of phenotypic and genetic variation within a
  predator population.
\newblock \emph{Ecology}, 92:\penalty0 1582--1593, 2011.

\bibitem[Takeuchi and Adachi(1983)]{takeuchi-adachi-83}
Y.~Takeuchi and N.~Adachi.
\newblock Existence and bifurcation of stable equilibrium in two-prey,
  one-predator communities.
\newblock \emph{Bulletin of Mathematical Biology}, 45\penalty0 (6):\penalty0
  877--900, 1983.

\bibitem[Thompson(1999)]{thompson-99}
J.N. Thompson.
\newblock {The evolution of species interactions.}
\newblock \emph{Science (New York, N.Y.)}, 284:\penalty0 2116--2118, 1999.

\bibitem[Tompkins et~al.(2000)Tompkins, Draycott, and Hudson]{tompkins-etal-00}
D.M. Tompkins, R.A.H. Draycott, and P.J. Hudson.
\newblock Field evidence for apparent competition mediated via the shared
  parasites of two gamebird species.
\newblock \emph{Ecology Letters}, 3:\penalty0 10--14, 2000.

\bibitem[Turelli and Barton(1994)]{turelli-barton-94}
M.~Turelli and N.~H. Barton.
\newblock Genetic and statistical analyses of strong selection on polygenic
  traits: what, me normal?
\newblock \emph{Genetics}, 138:\penalty0 913--941, 1994.

\bibitem[Vasseur et~al.(2011)Vasseur, Amarasekare, Rudolf, and
  Levine]{vasseur-etal-11}
D.A. Vasseur, P.~Amarasekare, V.H.W. Rudolf, and J.M. Levine.
\newblock {Eco-Evolutionary dynamics enable coexistence via neighbor-dependent
  selection.}
\newblock \emph{The American Naturalist}, 178:\penalty0 E96--E109, 2011.

\bibitem[Wilson and Turelli(1986)]{wilson-turelli-86}
D.S. Wilson and M.~Turelli.
\newblock Stable underdominance and the evolutionary invasion of empty niches.
\newblock \emph{American Naturalist}, 127:\penalty0 835--850, 1986.

\end{thebibliography}

\end{document}